\documentclass[onecolumn,amsmath,amssymb,10pt,aps]{revtex4}
  
\usepackage[utf8]{inputenc}
\usepackage[T1]{fontenc}

\usepackage{amsmath}
\usepackage{amsfonts}
\usepackage{graphicx}
\usepackage{yfonts}
\usepackage{color}
\usepackage[normalem]{ulem}
\usepackage{amsthm}
\usepackage{bm}
\usepackage{bbm}
\usepackage{mathtools}
\usepackage{array}
\usepackage{placeins}
\usepackage{enumitem}
\usepackage{tikz}
\usepackage{multirow}

\def\<{\langle}
\def\>{\rangle}

\newcommand{\Tr}{\mathrm{Tr}}
\def\oper{{\mathchoice{\rm 1\mskip-4mu l}{\rm 1\mskip-4mu l}
{\rm 1\mskip-4.5mu l}{\rm 1\mskip-5mu l}}}
\DeclareMathAlphabet\mathbfcal{OMS}{cmsy}{b}{n}

\mathchardef\mhyphen="2D 

\newtheorem{Definition}{Definition}

\newtheorem{Proposition}{Proposition}
\newtheorem{Example}{Example}

\begin{document}

\title{Entanglement witnesses and separability criteria based on generalized equiangular tight frames}

\author{Katarzyna Siudzi\'{n}ska$^\ast$}
\affiliation{Institute of Physics, Faculty of Physics, Astronomy and Informatics \\  Nicolaus Copernicus University in Toru\'{n}, ul. Grudzi\k{a}dzka 5, 87--100 Toru\'{n}, Poland \\ e-mail:kasias@umk.pl}

\begin{abstract}
In quantum information theory, measurements are represented by positive, operator-valued measures. In this paper, we use the operators corresponding to generalized equiangular measurements to construct positive maps. Their positivity follows from the properties of index of coincidence for few equiangular tight frames. These maps give rise to entanglement witnesses, which include as special cases many important classes considered in the literature. Additionally, we introduce separability criteria based on the correlation matrix and analyze them for various types of measurements.
\end{abstract}

\flushbottom

\maketitle

\thispagestyle{empty}

\section{Introduction}

Entanglement is one of the most intriguing and fundamental features of quantum mechanics, playing a central role in quantum information science, including quantum computation, quantum cryptography, and teleportation protocols \cite{HHHH,Nielsen}. However, detecting whether a given quantum state is entangled or separable can be a highly nontrivial task, especially in higher-dimensional or multipartite systems \cite{Gurvits}.

The Peres-Horodecki criterion, also known as the PPT (positive partial transpose) criterion, is a widely used method for detecting quantum entanglement in bipartite systems. It was introduced by Peres \cite{PhysRevLett.77.1413} and later fully characterized by the Horodecki family \cite{PHHH,HORODECKI1997333}. This criterion provides a necessary and sufficient separability condition for qubit-qubit and qubit-qutrit systems. In other instances, only some PPT states are separable. The rest are PPT entangled states, describing bound entanglement \cite{HHH_bound}. Therefore, more refined detection methods are needed. The CCNR (computable cross-norm or realignment) criterion \cite{Rudolph,HHH_OSID,ChenWu,ChenWu2} rearranges the density matrix via realignment and computes its trace norm. The nonlinear entanglement criteria use nonlinear functionals of the density matrix, like the purity or Tsallis entropy \cite{GuhneLev}. The correlation matrix criterion \cite{deVicente,deVicente2} and the covariance matrix criterion \cite{CMM1,Gittsovich} rely on calculations upon matrices constructed for given states.

Entanglement witnesses provide a powerful and practical tool for identifying entangled states. A Hermitian operator is called an entanglement witness if it is positive on all separable states and negative on at least one entangled state. While a single entanglement witness cannot detect all entangled states, every entangled state can, in principle, be detected by some witness \cite{Terhal1,Terhal2}. Importantly, entanglement detection can be carried out without requiring full knowledge of the state. Thus, entanglement witnesses form a central component of experimental and theoretical strategies for detecting quantum entanglement \cite{Barbieri,Lewenstein}.

Of particular interest are indecomposable entanglement witnesses, which allow one to detect PPT entangled states. Such states are also called bound entangled \cite{HHH_bound,HORODECKI1997333,bound_review}. They are non-distillable, which means that no protocols exist that transform them into the states that are maximally entangled \cite{Clarisse,Horodeccy2}. Bound entangled states find important applications e.g. in quantum cryptography \cite{Ozols,Augusiak} and quantum metrology \cite{Toth,Vertesi}. Moreover, bound entanglement can be \emph{activated} to provide resources important for quantum information and communication technologies \cite{Smolin,Kaneda,Ozaydin}.

In the construction of entanglement witnesses, one often uses positive but not completely positive maps defined via measurement operators. A wide class of witnesses from mutually unbiased bases (MUBs) \cite{Durt} was derived in ref. \cite{MUBs} and analyzed in more details for qutrit-qutrit systems. Further generalizations of this formalism include replacing the mutually unbiased bases with mutually unbiased measurements (MUMs) \cite{Li,P_maps}, symmetric informationally complete measurements (SIC POVMs) \cite{EW-SIC}, symmetric measurements \cite{SIC-MUB}, and equiangular tight frames \cite{W_ETS}. Witnesses for composite systems of unequal dimensions were also considered \cite{EW-2MUB}. The number of measurement operators required to detect bound entanglement was analyzed in refs. \cite{bound_ent,MUM_purity}.

In this paper, we provide a construction method for positive (but not trace preserving) maps and the associated entanglement witnesses from informationally overcomplete measurements. For this purpose, we apply the generalized equiangular measurements (GEAMs) \cite{GEAM}, which are sets of mutually unbiased non-projective equiangular tight frames. Importantly, we focus on those measurements that form conical 2-designs, because only then a partial index of coincidence can be derived.
Interestingly, the entanglement witnesses from GEAMs are parametrized only by two partial indices of coincidence.
Examples of indecomposable witnesses (detecting bound entanglement) are presented for the qutrit-qutrit scenario.
Finally, as an alternative way to detect entanglement, we provide enhanced separability criteria via the correlation matrix.

\section{Measurements from generalized equiangular tight frames}

A set $\{P_k:\,k=1,\ldots,M\}$ of rank-1 projectors $P_k$ acting on the Hilbert space $\mathcal{H}\simeq\mathbb{C}^d$ is an equiangular tight frame if $\Tr(P_kP_\ell)=c\Tr(P_k)\Tr(P_\ell)$ for all $k\neq\ell$ and $\sum_{k=1}^MP_k=\gamma\mathbb{I}_d$ with $\gamma>0$ \cite{Strohmer}. If $P_k$ are no longer rank-1, then they form a generalized equiangular tight frame \cite{GEAM}. Now, take collections of $N$ generalized equiangular tight frames $\mathcal{P}_\alpha=\{P_{\alpha,k}:\,k=1,\ldots,M_\alpha\}$, $\alpha=1,\ldots,N$, whose elements sum up to $\sum_{k=1}^{M_\alpha}P_{\alpha,k}=\gamma_\alpha\mathbb{I}_d$, where $\gamma_\alpha$ is a probability distribution. Additionally, assume that they are complementary \cite{PetzRuppert}, which means that $\Tr(P_{\alpha,k}P_{\beta,\ell})=f\Tr(P_{\alpha,k})\Tr(P_{\beta,\ell})$, $\alpha\neq\beta$. Under these conditions, one defines the corresponding positive, operator-valued measure (POVM).

\begin{Definition}\label{geam}(\cite{GEAM})
A collection of $N$ generalized equiangular tight frames $\{P_{\alpha,k}:\,k=1,\ldots,M_\alpha;\,\alpha=1,\ldots,N\}$ satisfying $\sum_{k=1}^{M_\alpha}P_{\alpha,k}=\gamma_\alpha\mathbb{I}_d$ for a probability distribution $\gamma_\alpha>0$ is a generalized equiangular measurement (GEAM) provided that $\sum_{\alpha=1}^NM_\alpha=d^2+N-1$ and
\begin{equation}
\begin{split}
\Tr(P_{\alpha,k})&=a_\alpha,\\
\Tr(P_{\alpha,k}^2)&=b_\alpha \Tr(P_{\alpha,k})^2,\\
\Tr(P_{\alpha,k}P_{\alpha,\ell})&=c_\alpha
\Tr(P_{\alpha,k})\Tr(P_{\alpha,\ell}),\qquad k\neq\ell,\\
\Tr(P_{\alpha,k}P_{\beta,\ell})&=
f\Tr(P_{\alpha,k})\Tr(P_{\beta,\ell}),\qquad \alpha\neq\beta,
\end{split}
\end{equation}
where
\begin{equation}
a_\alpha=\frac{d\gamma_\alpha}{M_\alpha},\qquad c_\alpha=\frac{M_\alpha-db_\alpha}{d(M_\alpha-1)},
\qquad f=\frac 1d,
\end{equation}
\begin{equation}\label{ba}
\frac 1d <b_\alpha\leq\frac 1d \min\{d,M_\alpha\}.
\end{equation}
\end{Definition}

Generalized equiangular measurements can be constructed from Hermitian orthonormal bases $\{G_0=\mathbb{I}_d/\sqrt{d},G_{\alpha,k}:\,k=1,\ldots,M_\alpha-1;\,\alpha=1,\ldots,N\}$ with traceless operators $G_{\alpha,k}$. Namely, one has \cite{GEAM}
\begin{equation}
P_{\alpha,k}=\frac{a_\alpha}{d}\mathbb{I}_d+\tau_\alpha H_{\alpha,k},
\end{equation}
where the traceless operators $H_{\alpha,k}$ are related to $G_{\alpha,k}$ in the following way,
\begin{equation}\label{H}
H_{\alpha,k}=\left\{\begin{aligned}
&G_\alpha-\sqrt{M_\alpha}(1+\sqrt{M_\alpha})G_{\alpha,k},\quad k=1,\ldots,M_\alpha-1,\\
&(1+\sqrt{M_\alpha})G_\alpha,\qquad k=M_\alpha,
\end{aligned}\right.
\end{equation}
with $G_\alpha=\sum_{k=1}^{M_\alpha-1}G_{\alpha,k}$. The real parameter
\begin{equation}
\tau_\alpha=\pm\sqrt{\frac{a_\alpha^2(b_\alpha-c_\alpha)}{M_\alpha(\sqrt{M_\alpha}+1)^2}}
\end{equation}
is chosen to guarantee that $P_{\alpha,k}\geq 0$.

There exists an important class of generalized equiangular measurements that are conical 2-designs \cite{Graydon,Graydon2} -- that is, they satisfy
\begin{equation}\label{con}
\sum_{\alpha=1}^N\sum_{k=1}^{M_\alpha}P_{\alpha,k}\otimes P_{\alpha,k}=
\kappa_+\mathbb{I}_d\otimes\mathbb{I}_d+\kappa_-\mathbb{F}_d
\end{equation}
for real parameters $\kappa_+\geq\kappa_->0$, where $\mathbb{F}_d=\sum_{m,n=0}^{d-1}|m\>\<n|\otimes|n\>\<m|$ is the flip operator. This equation holds as long as the GEAM parameters satisfy $a_\alpha^2(b_\alpha-c_\alpha)=S$ with
\begin{equation}\label{Srange}
0<S\leq \min\left\{\frac{d\gamma_\alpha^2}{M_\alpha},\frac{d-1}{M_\alpha-1}\frac{d\gamma_\alpha^2}{M_\alpha}\right\}.
\end{equation}
In this case, $\kappa_\pm$ in eq. (\ref{con}) are given by \cite{GEAM}
\begin{equation}\label{kappas2}
\kappa_+=\mu_N-\frac{S}{d},\qquad\kappa_-=S,\qquad \mu_N=\frac 1d \sum_{\alpha=1}^Na_\alpha\gamma_\alpha.
\end{equation}
It has been shown that a linear dependence between the purity $\Tr(\rho^2)$ of a state $\rho$ and its index of coincidence \cite{Rastegin5}
\begin{equation}\label{pak}
\mathcal{C}=\sum_{\alpha=1}^N\sum_{k=1}^{M_\alpha}p_{\alpha,k}^2,\qquad p_{\alpha,k}=\Tr(P_{\alpha,k}\rho),
\end{equation}
exists if and only if $P_{\alpha,k}$ form a conical 2-design.

\begin{Proposition}(\cite{GEAM}]
If the generalized equiangular measurement is a conical 2-design, then the associated index of coincidence satisfies
\begin{equation}\label{IOCN}
\mathcal{C}=S\left(\Tr\rho^2-\frac 1d\right)+\mu\leq\frac{d-1}{d}S+\mu,
\end{equation}
where $\mu=(1/d)\sum_{\alpha=1}^Na_\alpha\gamma_\alpha$.
\end{Proposition}

Interestingly, the above proposition can be further generalized. Observe that each generalized equiangular tight frame has its own index of coincidence (after rescaling)
\begin{equation}
\mathcal{C}_\alpha=\frac{1}{\gamma_\alpha^2}\sum_{k=1}^{M_\alpha}p_{\alpha,k}^2,\qquad 
\frac{1}{\gamma_\alpha}\sum_{k=1}^{M_\alpha}p_{\alpha,k}=1.
\end{equation}
Therefore, it is possible to derive an inequality for partial sums of squared probabilities, similar to that for few mutually unbiased bases \cite{Wu}.

\begin{Proposition}\label{prop}
If $\{P_{\alpha,k}:\,k=1,\ldots,M_\alpha;\,\alpha=1,\ldots,N\}$ is a conical 2-design, then the partial index of coincidence
\begin{equation}\label{IOC}
\mathcal{C}_L=\sum_{\alpha=1}^L\sum_{k=1}^{M_\alpha}p_{\alpha,k}^2\leq S\left(\Tr\rho^2-\frac 1d\right)+\mu_L,
\end{equation}
where $\mu_L=(1/d)\sum_{\alpha=1}^La_\alpha\gamma_\alpha$ and $1\leq L\leq N$. The equality is reached for $L=N$.
\end{Proposition}

\begin{proof}
In this proof, we follow the method introduced in ref. \cite{Rastegin2}.
Any quantum state $\rho$ can be represented as
\begin{equation}\label{rho}
\rho=\frac 1d \mathbb{I}_d+\sum_{\alpha=1}^N\sum_{k=1}^{M_\alpha}r_{\alpha,k}H_{\alpha,k}
\end{equation}
with real-valued parameters $r_{\alpha,k}$. Using the properties of the Hermitian operators $H_{\alpha,k}$,
\begin{equation}\label{HHH}
\begin{split}
\Tr(H_{\alpha,k}^2)&=(M_\alpha-1)(\sqrt{M_\alpha}+1)^2,\\
\Tr(H_{\alpha,k}H_{\alpha,\ell})&=-(\sqrt{M_\alpha}+1)^2,\qquad \ell\neq k,\\
\Tr(H_{\alpha,k}H_{\beta,\ell})&=0,\qquad \beta\neq\alpha.
\end{split}
\end{equation}
we calculate the probabilities
\begin{equation}
p_{\alpha,k}=\Tr(\rho P_{\alpha,k})
=\frac{a_\alpha}{d} +\tau_\alpha\Tr(\rho H_{\alpha,k})
=\frac{a_\alpha}{d}+\tau_\alpha\sum_{\ell=1}^{M_\alpha}r_{\alpha,\ell}\Tr(H_{\alpha,k}H_{\alpha,\ell})
=\frac{a_\alpha}{d} +\tau(\sqrt{M_\alpha}+1)^2(M_\alpha r_{\alpha,k}-r_\alpha),
\end{equation}
where $r_\alpha=\sum_{k=1}^{M_\alpha}r_{\alpha,k}$. Now, the sum of squares over the second index satisfies
\begin{equation}\label{sum}
\sum_{k=1}^{M_\alpha}p_{\alpha,k}^2
=\frac{a_\alpha^2 M_\alpha}{d^2}+\tau_\alpha^2M_\alpha(\sqrt{M_\alpha}+1)^4
\left(M_\alpha\sum_{k=1}^{M_\alpha} r_{\alpha,k}^2-r_\alpha^2\right).
\end{equation}
Remembering that $\sum_{k=1}^{M_\alpha}p_{\alpha,k}=\gamma_\alpha$, one finds
\begin{equation}
\sum_{k=1}^{M_\alpha}p_{\alpha,k}^2\geq
\frac{\gamma_\alpha^2}{M_\alpha}=\frac{a_\alpha^2 M_\alpha}{d^2},
\end{equation}
where we used the fact that $\sum_{k=1}^{M_\alpha}p_{\alpha,k}^2$ reaches its lower bound when all $p_{\alpha,k}=\gamma_\alpha/M_\alpha$.
This shows that the second term on the right hand-side of eq. (\ref{sum}) is positive. Now, we compute the sum over an incomplete set of indices $\alpha=1,\ldots,L$;
\begin{equation}\label{almost}
\begin{split}
\sum_{\alpha=1}^L\sum_{k=1}^{M_\alpha}p_{\alpha,k}^2
&=\frac{1}{d^2} \sum_{\alpha=1}^L a_\alpha^2 M_\alpha +\sum_{\alpha=1}^L\tau_\alpha^2M_\alpha(\sqrt{M_\alpha}+1)^4
\left(M_\alpha\sum_{k=1}^{M_\alpha} r_{\alpha,k}^2-r_\alpha^2\right)\\
&=\frac{1}{d^2} \sum_{\alpha=1}^L a_\alpha^2 M_\alpha +\sum_{\alpha=1}^La_\alpha^2(b_\alpha-c_\alpha)(\sqrt{M_\alpha}+1)^2
\left(M_\alpha\sum_{k=1}^{M_\alpha} r_{\alpha,k}^2-r_\alpha^2\right).
\end{split}
\end{equation}
On the other hand, from eq. (\ref{rho}), we find the purity of the state $\rho$ is equal to
\begin{equation}\label{rho2}
\Tr(\rho^2)=\frac 1d + \sum_{\alpha=1}^N(\sqrt{M_\alpha}+1)^2
\left(M_\alpha\sum_{k=1}^{M_\alpha} r_{\alpha,k}^2-r_\alpha^2\right),
\end{equation}
where we used properties (\ref{HHH}) to compute
\begin{equation}
\sum_{k,\ell=1}^{M_\alpha}r_{\alpha,k}r_{\alpha,\ell}\Tr(H_{\alpha,k}H_{\alpha,\ell})
=(\sqrt{M_\alpha}+1)^2\left(M_\alpha\sum_{k=1}^{M_\alpha}r_{\alpha,k}^2-r_\alpha^2\right),
\end{equation}
After comparing eqs. (\ref{almost}) and (\ref{rho2}), it becomes evident that the relation between the sum of squared $p_{\alpha,k}$ and the purity of $\rho$ is linear provided that $a_\alpha^2(b_\alpha-c_\alpha)=S$, which is exactly the condition for a conical 2-design from Proposition 1. In this case,
\begin{equation}
\begin{split}
\sum_{\alpha=1}^L\sum_{k=1}^{M_\alpha}p_{\alpha,k}^2
&=\frac{1}{d^2} \sum_{\alpha=1}^L a_\alpha^2 M_\alpha 
+S\sum_{\alpha=1}^L(\sqrt{M_\alpha}+1)^2
\left(M_\alpha\sum_{k=1}^{M_\alpha} r_{\alpha,k}^2-r_\alpha^2\right)\\
&\leq\frac{1}{d^2} \sum_{\alpha=1}^L a_\alpha^2 M_\alpha 
+S\sum_{\alpha=1}^N(\sqrt{M_\alpha}+1)^2
\left(M_\alpha\sum_{k=1}^{M_\alpha} r_{\alpha,k}^2-r_\alpha^2\right)\\
&=S\left[\Tr(\rho^2)-\frac 1d\right]+\mu_L,
\end{split}
\end{equation}
where $\mu_L=\sum_{\alpha=1}^L a_\alpha^2 M_\alpha/d^2=(1/d)\sum_{\alpha=1}^La_\alpha\gamma_\alpha$. For $L=N$, this result reproduces the index of coincidence $\mathcal{C}$, for which the equality holds \cite{GEAM}.
\end{proof}

Observe that eq. (\ref{IOC}) is upper bounded by its value on pure states,
\begin{equation}\label{IOC2}
\mathcal{C}_L\leq \widetilde{\mathcal{C}}_L=\frac{d-1}{d}S+\mu_L,
\end{equation}
where $\Tr(\rho^2)=1$.

\section{Positive maps and entanglement witnesses}

Using generalized equiangular measurements that are conical 2-designs, we construct $N$ maps
\begin{equation}
\Phi_\alpha[X]=\sum_{k,\ell=1}^{M_\alpha}\mathcal{O}^{(\alpha)}_{k\ell}P_{\alpha,k}\Tr(XP_{\alpha,\ell}),
\end{equation}
where $\mathcal{O}^{(\alpha)}$ are orthogonal rotation matrices preserving the maximally mixed vector $\mathbf{n}_\ast=(1,\ldots,1)/\sqrt{d}$. Note that
\begin{equation}
\Tr(\Phi_\alpha[X])=a_\alpha\gamma_\alpha\Tr(X),
\end{equation}
which means that $\Phi_\alpha$ are not trace preserving. This is important in proving the following proposition, where trace-preserving $\Phi_\alpha$ would result in imposing additional conditions on $P_{\alpha,k}$.

\begin{Proposition}\label{Prop}
The linear map
\begin{equation}\label{PTP}
\Phi=A\Phi_0+\sum_{\alpha=L+1}^K\Phi_\alpha-\sum_{\alpha=1}^L\Phi_\alpha
\end{equation}
with $A=d(2\widetilde{\mathcal{C}}_L-\widetilde{\mathcal{C}}_K)$ and $\Phi_0[X]=\mathbb{I}_d\Tr(X)/d$ is positive, where $1\leq L\leq K\leq N$.
\end{Proposition}

\begin{proof}
The condition for positivity of $\Phi$ follows directly from Mehta's theorem \cite{Mehta}
\begin{equation}\label{pos}
\frac{\Tr(\Phi[P])^2}{[\Tr(\Phi[P])]^2}\leq\frac{1}{d-1}.
\end{equation}
It is enough to prove the above inequality for rank-1 projectors $P$ \cite{MUBs}. The denominator simply reads
\begin{equation}
[\Tr(\Phi[P])]^2=(A+d\mu_K-2d\mu_L)^2.
\end{equation}
Computing the numerator requires more attention. First, note that
\begin{equation}\label{PTP1}
\begin{split}
\Tr(\Phi[P])^2=\Tr\Bigg\{&A^2\Phi_0[P]^2+\sum_{\alpha,\beta=1}^L
\Phi_\alpha[P]\Phi_\beta[P]+\sum_{\alpha,\beta=L+1}^K\Phi_\alpha[P]\Phi_\beta[P]
+2A\sum_{\alpha=L+1}^K\Phi_0[P]\Phi_\alpha[P]\\&-2A\sum_{\alpha=1}^L\Phi_0[P]\Phi_\alpha[P]
-2\sum_{\alpha=L+1}^K\sum_{\beta=1}^L\Phi_\alpha[P]\Phi_\beta[P]\Bigg\}.
\end{split}
\end{equation}
Next, using the properties of the rotation matrix
\begin{equation}
\sum_{k=1}^{M_\alpha}\mathcal{O}^{(\alpha)}_{k\ell}=\sum_{\ell=1}^{M_\alpha}\mathcal{O}^{(\alpha)}_{k\ell}=1,\qquad
\sum_{k=1}^{M_\alpha}\mathcal{O}^{(\alpha)}_{k\ell}\mathcal{O}^{(\alpha)}_{km}=\delta_{\ell m},
\end{equation}
as well as the properties of $\Phi_\alpha$,
\begin{equation}
\Tr(\Phi_0[P]^2)=\frac 1d,\qquad \Tr(\Phi_0[P]\Phi_\alpha[P])=\frac{a_\alpha\gamma_\alpha}{d},\qquad
\Tr(\Phi_\alpha[P]\Phi_\beta[P])=\frac{a_\alpha a_\beta\gamma_\alpha\gamma_\beta}{d},\quad \alpha\neq\beta,
\end{equation}
we simplify eq. (\ref{PTP1}) to
\begin{equation}\label{PTP2}
\begin{split}
d\Tr(\Phi[P])^2=(A+d\mu_K-2d\mu_L)^2-\sum_{\alpha=1}^Ka_\alpha^2\gamma_\alpha^2
+d\sum_{\alpha=1}^K\Tr(\Phi_\alpha[P])^2,
\end{split}
\end{equation}
where $\mu_L=\sum_{\alpha=1}^La_\alpha\gamma_\alpha/d$, just as defined in Proposition 1. Now, observe that
\begin{equation}
\Tr(\Phi_\alpha[P]^2)=a_\alpha^2\gamma_\alpha^2c_\alpha+a_\alpha^2(b_\alpha-c_\alpha)
\sum_{k=1}^{M_\alpha}p_{\alpha,k}^2
\end{equation}
with $p_{\alpha,k}=\Tr(PP_{\alpha,k})$ being calculated on a pure state $P$. In this step, it becomes necessary to assume that $P_{\alpha,k}$ form a conical 2-design (i.e., $a_\alpha^2(b_\alpha-c_\alpha)=S$). Then, we can apply the results from eq. (\ref{IOC2}). For pure states $\rho=P$, the purity $\Tr(\rho^2)=1$, and hence
\begin{equation}
\sum_{\alpha=1}^K\Tr(\Phi_\alpha[P]^2)=\sum_{\alpha=1}^Ka_\alpha^2\gamma_\alpha^2c_\alpha
+S\sum_{\alpha=1}^K\sum_{k=1}^{M_\alpha}p_{\alpha,k}^2\leq
\sum_{\alpha=1}^Ka_\alpha^2\gamma_\alpha^2c_\alpha+\frac{d-1}{d}S^2+S\mu_K.
\end{equation}
This way, eq. (\ref{PTP2}) reduces to
\begin{equation}\label{PTP4}
d\Tr(\Phi[P])^2\leq (A+d\mu_K-2d\mu_L)^2
-\sum_{\alpha=1}^Ka_\alpha^2\gamma_\alpha^2(1-dc_\alpha)+S[(d-1)S+d\mu_K].
\end{equation}
Finally, note that
\begin{equation}
S=a_\alpha^2(b_\alpha-c_\alpha)=a_\alpha\gamma_\alpha(1-dc_\alpha),
\end{equation}
and therefore
\begin{equation}\label{PTP3}
d\Tr(\Phi[P])^2\leq (A+d\mu_K-2d\mu_L)^2+(d-1)S^2.
\end{equation}
Now, if we recall the definition of $A$ and $\widetilde{\mathcal{C}}_L$, we see that
\begin{equation}
A=d(2\widetilde{\mathcal{C}}_L-\widetilde{\mathcal{C}}_K)=(d-1)S-d(\mu_K-2\mu_L),
\end{equation}
and eq. (\ref{PTP3}) further simplifies to
\begin{equation}
d\Tr(\Phi[P])^2\leq d(d-1)S^2=\frac{d[\Tr(\Phi[P])]^2}{d-1},
\end{equation}
which proves that $\Phi$ is indeed a positive map.
\end{proof}

Positive but not completely positive maps are used in the theory of quantum entanglement to construct block-positive operators that detect entangled (non-separable) states. Through the Choi-Jamio{\l}kowski isomorphism \cite{Choi,Jamiolkowski}
\begin{equation}\label{W}
W=\sum_{m,n=0}^{d-1}|m\>\<n|\otimes\Phi[|m\>\<n|],
\end{equation}
a one-to-one correspondence is established between a positive map $\Phi$ and an entanglement witness $W$, using an operational basis $|k\>$ in $\mathbb{C}^d$. 
In particular, the witness corresponding to $\Phi$ from Proposition 2 reads
\begin{equation}\label{W2}
W=\Big(2\widetilde{\mathcal{C}}_L-\widetilde{\mathcal{C}}_K\Big)\mathbb{I}_d\otimes\mathbb{I}_d
+\sum_{\alpha=L+1}^KJ_\alpha-\sum_{\alpha=1}^LJ_\alpha
\end{equation}
where
\begin{equation}
J_\alpha=\sum_{k,\ell=1}^{M_\alpha}\mathcal{O}_{k\ell}^{(\alpha)}
\overline{P}_{\alpha,\ell}\otimes P_{\alpha,k}.
\end{equation}
Note that $W$ is parametrized only by two partial indices of coincidence: one associated with the number of negative terms $L$ and the other with the number of all terms $K$.

By choosing specific generalized equiangular measurement operators $P_{\alpha,k}$ and the rotation matrices $\mathcal{O}_{k\ell}^{(\alpha)}$, one finds explicit forms of entanglement witnesses. Note that the construction method of $\Phi$ from Proposition \ref{Prop} allows also for completely positive maps. However, $W$ corresponding to a completely positive $\Phi$ does not detect quantum entanglement due to $\Tr(W\rho)\geq 0$ for all $\rho$. Obtaining analytical formulas for complete positivity of a general $\Phi$ is a highly non-trivial task, therefore it is crucial to check for it after constructing explicit examples.

\begin{Proposition}
A trivial case follows from the choice $\mathcal{O}^{(\alpha)}=\mathbb{I}_{M_\alpha}$ and $L=K=N$. Then, using eq. (\ref{con}), one gets
\begin{equation}
W=S\left(\mathbb{I}_{d^2}-dP_+\right),
\end{equation}
where $P_+=(1/d)\sum_{m,n=0}^{d-1}|m\>\<n|\otimes|m\>\<n|$ is a maximally entangled state. This is the entanglement witness for the reduction map \cite{EW_reduction,P_maps}.
\end{Proposition}

\section{Detection of bound entanglement}

Of special interest are entanglement witnesses that detect bound entangled states -- that is, entangled states that cannot be distilled into maximally entangled states using local operations and classical communication (LOCC). Such witnesses are known as indecomposable or non-decomposable. On the other hand, there are decomposable witnesses, which can be represented as $W=X+Y^\Gamma$, where $X$ and $Y$ are positive operators, and $\Gamma=\oper\otimes T$ is a partial transposition. Construction of indecomposable entanglement witnesses is non-trivial but important for the purpose of detecting PPT (positive partial transposition) entangled quantum states.

In what follows, we provide qutrit ($d=3$) examples of indecomposable entanglement witnesses from generalized equiangular measurements that go beyond the well-known classes of equinumerous measurements. The GEAMs are constructed as in ref. \cite{GEAM}, using the normalized Gell-Mann matrices: $G_{0}=\mathbb{I}_3/\sqrt{3}$ and
\begin{align*}
G_{1,1}=\frac{1}{\sqrt{2}}
\begin{pmatrix}
0 & 1 & 0 \\
1 & 0 & 0 \\
0 & 0 & 0
\end{pmatrix},\qquad
G_{1,2}=\frac{1}{\sqrt{2}}
\begin{pmatrix}
0 & -i & 0 \\
i & 0 & 0 \\
0 & 0 & 0
\end{pmatrix},\qquad
G_{2,1}=\frac{1}{\sqrt{2}}
\begin{pmatrix}
0 & 0 & 1 \\
0 & 0 & 0 \\
1 & 0 & 0
\end{pmatrix},\qquad
&G_{2,2}=\frac{1}{\sqrt{2}}
\begin{pmatrix}
0 & 0 & -i \\
0 & 0 & 0 \\
i & 0 & 0
\end{pmatrix},\\
G_{3,1}=\frac{1}{\sqrt{2}}
\begin{pmatrix}
0 & 0 & 0 \\
0 & 0 & 1 \\
0 & 1 & 0
\end{pmatrix},\qquad
G_{3,2}=\frac{1}{\sqrt{2}}
\begin{pmatrix}
0 & 0 & 0 \\
0 & 0 & -i \\
0 & i & 0
\end{pmatrix},\qquad
G_{3,3}=\frac{1}{\sqrt{2}}
\begin{pmatrix}
1 & 0 & 0 \\
0 & -1 & 0 \\
0 & 0 & 0
\end{pmatrix},\qquad
&G_{3,4}=\frac{1}{\sqrt{6}}
\begin{pmatrix}
1 & 0 & 0 \\
0 & 1 & 0 \\
0 & 0 & -2
\end{pmatrix}.
\end{align*}
For the first two POVMs ($\alpha=1,2$), we choose the mutually unbiased measurements \cite{Kalev}
\begin{equation}\label{E12}
E_{\alpha,k}=\left\{\begin{aligned}
&\frac 13 \mathbb{I}_3+t_\alpha\left[G_{1,1}+G_{1,2}-\sqrt{3}(1+\sqrt{3})G_{\alpha,k}
\right],\quad k=1,2,\\
&\frac 13 \mathbb{I}_3+t_\alpha(1+\sqrt{3})(G_{1,1}+G_{1,2}),\qquad k=3.
\end{aligned}\right.
\end{equation}
The last POVM is a $(1,5)$-POVM \cite{SIC-MUB} with
\begin{equation}\label{E3}
E_{3,k}=\left\{\begin{aligned}
&\frac 15 \mathbb{I}_3+t_3
\left[G_{3,1}+G_{3,2}+G_{3,3}+G_{3,4}-\sqrt{5}(1+\sqrt{5})G_{\alpha,k}
\right],\quad k=1,2,3,4,\\
&\frac 15 \mathbb{I}_3+t_3(1+\sqrt{5})(G_{3,1}+G_{3,2}+G_{3,3}+G_{3,4}),\qquad k=5,
\end{aligned}\right.
\end{equation}
The parameters $t_\alpha$ are chosen to be non-optimal -- that is, there exist greater values of $t_\alpha$ that guarantee $E_{\alpha,k}\geq 0$. To simplify further calculations, we take 
\begin{equation}
t_1=t_2=\frac{1}{3(1+\sqrt{3})},\qquad t_3=\frac{1}{10(1+\sqrt{5})}.
\end{equation}
The corresponding generalized equiangular measurement follows from
\begin{equation}
P_{\alpha,k}=\gamma_\alpha E_{\alpha,k},
\end{equation}
where $\gamma_1=\gamma_2$ and $\gamma_3=1-2\gamma_1$ with
\begin{equation}\label{gamma1}
\gamma_1=\sqrt{5}-2.
\end{equation}
By construction, this GEAM is a conical 2-design with $S=9/4-\sqrt{5}$.

The most general $3\times 3$ rotation matrix that preserves the maximally mixed vector is the cyclic matrix \cite{MUBs}
\begin{equation}
\mathcal{R}(\theta)=
\begin{pmatrix}
c_1(\theta) & c_2(\theta) & c_3(\theta) \\
c_3(\theta) & c_1(\theta) & c_2(\theta) \\
c_2(\theta) & c_3(\theta) & c_1(\theta)
\end{pmatrix}
\end{equation}
with
\begin{equation}
c_1(\theta)=\frac 13 + \frac 23 \cos(\theta),\qquad c_2(\theta)=\frac 13 + \frac 23 \cos(\theta-2\pi/3),\qquad c_3(\theta)=\frac 13 + \frac 23 \cos(\theta+2\pi/3).
\end{equation}
In the following examples, we take
\begin{equation}
\mathcal{O}^{(1)}=\mathcal{R}(2\pi/3),\qquad 
\mathcal{O}^{(2)}=\mathcal{R}(-2\pi/3),
\end{equation}
and the $5\times 5$ permutation matrix
\begin{equation}
\mathcal{O}^{(3)}=
\begin{pmatrix}
0 & 1 & 0 & 0 & 0 \\
1 & 0 & 0 & 0 & 0 \\
0 & 0 & 0 & 1 & 0 \\
0 & 0 & 1 & 0 & 0 \\
0 & 0 & 0 & 0 & 1
\end{pmatrix}.
\end{equation}

Examples of PPT states $\rho$ detected via a given witness $W$ are constructed by demanding that $\rho$ has the same form as $W$. Then, conditions $\rho\geq 0$ and $\rho^\Gamma\geq 0$ are applied, and the remaining free parameters are chosen so that $\Tr(\rho W)<0$.

\begin{Example}
For $K=N=3$ and $L=2$, we construct the entanglement witness
\begin{equation}
W=(2\widetilde{C}_2-\widetilde{C}_3)\mathbb{I}_3\otimes\mathbb{I}_3+J_1-J_2-J_3,
\end{equation}
where
\begin{equation}
\widetilde{C}_2=\frac 23 S+\frac{\gamma_2^2}{3}+\frac{\gamma_3^2}{5}.
\end{equation}
The explicit form of $W$ reads
\begin{equation}
W=
\left[\begin{array}{c c c|c c c|c c c}
x & \cdot & \cdot & \cdot & iu & \cdot & \cdot & \cdot & \cdot \\
\cdot & y & \cdot & \cdot & \cdot & \cdot & \cdot & \cdot & \cdot \\
\cdot & \cdot & z & \cdot & \cdot & \cdot & \cdot & \cdot & \cdot \\
\hline
\cdot & \cdot & \cdot & y & \cdot & \cdot & \cdot & \cdot & \cdot \\
-iu & \cdot & \cdot & \cdot & z & \cdot & \cdot & \cdot & \cdot \\
\cdot & \cdot & \cdot & \cdot & \cdot & x & \cdot & -iv & \cdot \\
\hline
\cdot & \cdot & \cdot & \cdot & \cdot & \cdot & z & \cdot & \cdot \\
\cdot & \cdot & \cdot & \cdot & \cdot & iv & \cdot & x & \cdot \\
\cdot & \cdot & \cdot & \cdot & \cdot & \cdot & \cdot & \cdot & y
\end{array}\right]
\end{equation}
where
\begin{equation}\label{form2}
\begin{split}
&x=\frac{(2-\sqrt{3})(7-3\sqrt{5})}{12(3+\sqrt{5})},\qquad 
y=\frac 16 (9-4\sqrt{5}),\qquad 
z=\frac{(2+\sqrt{3})(7-3\sqrt{5})}{12(3+\sqrt{5})},\\
u&=\frac{\sqrt{3}}{4} (9-4\sqrt{5}),\qquad 
v=\frac{1}{4\sqrt{161+72\sqrt{5}}},
\end{split}
\end{equation}
and the dots represent zeros in the matrix entries. This witness is indecomposable as it detects the rank-8 PPT entangled state
\begin{equation}
\rho=\frac{1}{300}
\left[\begin{array}{c c c|c c c|c c c}
43 & \cdot & \cdot & \cdot & -2i\sqrt{258} & \cdot & \cdot & \cdot & \cdot \\
\cdot & 33 & \cdot & \cdot & \cdot & \cdot & \cdot & \cdot & \cdot \\
\cdot & \cdot & 24 & \cdot & \cdot & \cdot & \cdot & \cdot & \cdot \\
\hline
\cdot & \cdot & \cdot & 33 & \cdot & \cdot & \cdot & \cdot & \cdot \\
2i\sqrt{258} & \cdot & \cdot & \cdot & 24 & \cdot & \cdot & \cdot & \cdot \\
\cdot & \cdot & \cdot & \cdot & \cdot & 43 & \cdot & 6i\sqrt{22} & \cdot \\
\hline
\cdot & \cdot & \cdot & \cdot & \cdot & \cdot & 24 & \cdot & \cdot \\
\cdot & \cdot & \cdot & \cdot & \cdot & -6i\sqrt{22} & \cdot & 43 & \cdot \\
\cdot & \cdot & \cdot & \cdot & \cdot & \cdot & \cdot & \cdot & 33
\end{array}\right],
\end{equation}
which can be shown by calculating the expectation value
\begin{equation}
\Tr(\rho W)\simeq -2.22\cdot 10^{-5}.
\end{equation}
\end{Example}

\begin{Example}
This time, we fix $K=N=3$ and $L=1$. The entanglement witness is constructed as follows,
\begin{equation}
W=(2\widetilde{C}_1-\widetilde{C}_3)\mathbb{I}_3\otimes\mathbb{I}_3+J_1-J_2+J_3,
\end{equation}
where
\begin{equation}
\widetilde{C}_1=\frac 23 S+\frac{\gamma_2^2}{3}.
\end{equation}
We find that the matrix representation of this witness is of the form
\begin{equation}
W=
\left[\begin{array}{c c c|c c c|c c c}
z & \cdot & \cdot & \cdot & iu & \cdot & \cdot & \cdot & \cdot \\
\cdot & y & \cdot & \cdot & \cdot & \cdot & \cdot & \cdot & \cdot \\
\cdot & \cdot & x & \cdot & \cdot & \cdot & \cdot & \cdot & \cdot \\
\hline
\cdot & \cdot & \cdot & y & \cdot & \cdot & \cdot & \cdot & \cdot \\
-iu & \cdot & \cdot & \cdot & x & \cdot & \cdot & \cdot & \cdot \\
\cdot & \cdot & \cdot & \cdot & \cdot & z & \cdot & -iv & \cdot \\
\hline
\cdot & \cdot & \cdot & \cdot & \cdot & \cdot & x & \cdot & \cdot \\
\cdot & \cdot & \cdot & \cdot & \cdot & iv & \cdot & z & \cdot \\
\cdot & \cdot & \cdot & \cdot & \cdot & \cdot & \cdot & \cdot & y
\end{array}\right]
\end{equation}
with the parameters given by eq. (\ref{form2}).
Again, $W$ is indecomposable, allowing us to detect the rank-7 PPT entangled state
\begin{equation}
\rho=\frac{1}{750}
\left[\begin{array}{c c c|c c c|c c c}
31 & \cdot & \cdot & \cdot & -5i\sqrt{186} & \cdot & \cdot & \cdot & \cdot \\
\cdot & 69 & \cdot & \cdot & \cdot & \cdot & \cdot & \cdot & \cdot \\
\cdot & \cdot & 150 & \cdot & \cdot & \cdot & \cdot & \cdot & \cdot \\
\hline
\cdot & \cdot & \cdot & 69 & \cdot & \cdot & \cdot & \cdot & \cdot \\
5i\sqrt{186} & \cdot & \cdot & \cdot & 150 & \cdot & \cdot & \cdot & \cdot \\
\cdot & \cdot & \cdot & \cdot & \cdot & 31 & \cdot & -31i & \cdot \\
\hline
\cdot & \cdot & \cdot & \cdot & \cdot & \cdot & 150 & \cdot & \cdot \\
\cdot & \cdot & \cdot & \cdot & \cdot & 31i & \cdot & 31 & \cdot \\
\cdot & \cdot & \cdot & \cdot & \cdot & \cdot & \cdot & \cdot & 69
\end{array}\right].
\end{equation}
The expectation value of the witness $W$ in this state is equal to
\begin{equation}
\Tr(\rho W)\simeq -8.05\cdot 10^{-5}.
\end{equation}
\end{Example}

\section{Separability criteria}

Take a bipartite composite system with subsystems labeled by $A$ and $B$, so that the total state $\rho:\mathcal{H}_A\otimes\mathcal{H}_B\to\mathcal{H}_A\otimes\mathcal{H}_B$. On each subsystem, define a generalized equiangular measurement: $\{P_{\alpha,k}^{A}:\,k=1,\ldots,M_\alpha^A;\,\alpha=1,\ldots,N_A\}$ and analogically for $B$. No further assumptions are made regarding the measurement parameters $(a_\alpha,b_\alpha,c_\alpha,f)$, so in general they can depend on the choice of subsystem. Now, to measure correlations between subsystems $A$ and $B$, one introduces the correlation matrix \cite{deVicente,deVicente2} with elements
\begin{equation}
\mathcal{P}_{\alpha,k;\beta,\ell}=\Tr\Big[\rho(P_{\alpha,k}^{A}\otimes P_{\beta,\ell}^{B})\Big].
\end{equation}
The first separability condition follows in terms of its trace, and as it turns out it depends on the maximal indices of coincidence $\widetilde{\mathcal{C}}^A$, $\widetilde{\mathcal{C}}^B$ associated with $\{P_{\alpha,k}^{A}\}$ and $\{P_{\alpha,k}^{B}\}$, respectively.

\begin{Proposition}\label{esicA}
If a bipartite state $\rho$ is separable, then the necessary separability condition reads
\begin{equation}
\Tr\mathcal{P}\leq\frac{\widetilde{\mathcal{C}}^A+\widetilde{\mathcal{C}}^B}{2},
\end{equation}
where the GEAMs are chosen in such a way that $N_A=N_B$ and $M_\alpha^A=M_\alpha^B$.
\end{Proposition}

\begin{proof}
For a product state $\rho=\rho_A\otimes\rho_B$, one has
$\mathcal{P}_{\alpha,k;\beta,\ell}=p_{\alpha,k}^Ap_{\beta,\ell}^B$ with $p_{\alpha,k}^A=\Tr(P_{\alpha,k}^{A}\rho_A)$ and $p_{\alpha,k}^B=\Tr(P_{\alpha,k}^{B}\rho_B)$. 
Note that the trace of $\mathcal{P}$ is well defined only under the additional condition that $\mathcal{P}$ is a square matrix. This is the case as long as the GEAMs on both subsystems are of equal number: $N_A=N_B\equiv N$ and $M_\alpha^A=M_\alpha^B\equiv M_\alpha$. Therefore, under this additional condition,
\begin{equation}
\Tr\mathcal{P}=\sum_{\alpha=1}^N\sum_{k=1}^{M_\alpha}p_{\alpha,k}^Ap_{\alpha,k}^B
\leq\frac 12 \sum_{\alpha=1}^N\sum_{k=1}^{M_\alpha}\Big[(p_{\alpha,k}^A)^2+(p_{\alpha,k}^B)^2\Big]
=\frac{\mathcal{C}^A+\mathcal{C}^B}{2}\leq\frac{\widetilde{\mathcal{C}}^A+\widetilde{\mathcal{C}}^B}{2}.
\end{equation}
From the linearity of the trace, this result extends from product states to all separable states.
\end{proof}

Let us check the detection properties of this separability condition on the example of the Werner states \cite{Werner}
\begin{equation}
\rho(\phi)=\frac{1}{d(d^2-1)}\Big[(d-\phi)\mathbb{I}_{d^2}+(d\phi-1)\mathbb{F}_d\Big],
\end{equation}
where $-1\leq \phi\leq 1$ and $\mathbb{F}_d$ is the flip operator. These states are entangled for $\phi<0$ and separable for $\phi\geq 0$.
Assuming that the same GEAM $P_{\alpha,k}$ is used on both subsystems, one gets
\begin{equation}\label{trp}
\begin{split}
\Tr\mathcal{P}&=\sum_{\alpha=1}^N\sum_{k=1}^{M_\alpha}\mathcal{P}_{\alpha,k;\alpha,k}
=\sum_{\alpha=1}^N\frac{a_\alpha^2M_\alpha}{d(d^2-1)}[d-\phi+(d\phi-1)b_\alpha]\\
&=\frac{1}{d(d^2-1)}\left[d^2\mu(d-\phi)+(d\phi-1)\sum_{\alpha=1}^Na_\alpha^2b_\alpha M_\alpha\right].
\end{split}
\end{equation}
Now, remembering that $P_{\alpha,k}$ form a conical 2-design, we find that
\begin{equation}
S=a_\alpha^2(b_\alpha-c_\alpha)=a_\alpha^2\frac{M_\alpha(db_\alpha-1)}{d(M_\alpha-1)}.
\end{equation}
After multiplying both sides by $d(M_\alpha-1)$ and taking a sum over all $\alpha$, one arrives at
\begin{equation}
dS\sum_{\alpha=1}^N(M_\alpha-1)=\sum_{\alpha=1}^Na_\alpha^2M_\alpha(db_\alpha-1),
\end{equation}
which further reduces to
\begin{equation}\label{ssuma}
d(d^2-1)S=d\sum_{\alpha=1}^Na_\alpha^2b_\alpha M_\alpha-d^2\mu.
\end{equation}
Using the above equality in eq. (\ref{trp}) results in
\begin{equation}
\Tr\mathcal{P}=\mu+\frac{S}{d}(d\phi-1).
\end{equation}
Hence, the separability criterion from Proposition \ref{esicA} produces
\begin{equation}
\mu+\frac{S}{d}(d\phi-1)\leq \frac{d-1}{d}S+\mu,
\end{equation}
which is equivalent to $\phi\leq 1$. This means that any conical 2-design GEAM detects all separable Werner states but no entanglement.

Now, we take the isotropic states \cite{iso}
\begin{equation}
\rho_{\rm iso}(\theta)=\frac{1-\theta}{d^2-1}(\mathbb{I}_{d^2}-P_+)+\theta P_+,
\end{equation}
where $0\leq \theta\leq 1$, and $P_+=(1/d)\sum_{m,n=1}^d|m\>\<n|\otimes|m\>\<n|$ is a maximally entangled state. It is known that $\rho_{\rm iso}(\theta)$ are entangled for $\theta>1/d$ and separable for $\theta\leq 1/d$. Taking $P_{\alpha,k}^A=P_{\alpha,k}$ and $P_{\alpha,k}^B=\overline{P}_{\alpha,k}$ results in
\begin{equation}
\Tr\mathcal{P}=\sum_{\alpha=1}^N\frac{a_\alpha^2M_\alpha}{d(d^2-1)}\Big[d(1-\theta)
+(d^2\theta-1)b_\alpha\Big]=\mu+\frac{S}{d}(d^2\theta-1),
\end{equation}
where we used eq. (\ref{ssuma}) to simplify the formula. According to Proposition \ref{esicA}, $\rho_{\rm iso}(\theta)$ is separable if
\begin{equation}
\mu+\frac{S}{d}(d^2\theta-1)\leq\frac{d-1}{d}S+\mu,
\end{equation}
or $\theta\leq 1/d$. If $\theta>1/d$, this state is entangled. Hence, any conical 2-design GEAM detects all separable and entangled isotropic states.

Alternative separability conditions involve the trace norm $\|X\|_{\Tr}=\Tr\sqrt{X^\dagger X}$.

\begin{Proposition}\label{esicB}
If a bipartite state $\rho$ is separable, then $\|\mathcal{P}\|_{\Tr}\leq \sqrt{\widetilde{\mathcal{C}}^A\widetilde{\mathcal{C}}^B}$.
\end{Proposition}

\begin{proof}
Using the notation from the proof to Proposition \ref{esicA} and the methods from ref. \cite{ESIC}, we show that
\begin{equation}
\|\mathcal{P}\|_{{\rm{tr}}}=\Tr\sqrt{\mathcal{P}^2}=\sqrt{\sum_{\alpha=1}^{N_A}
\sum_{k=1}^{M_\alpha^A}(p_{\alpha,k}^A)^2}
\sqrt{\sum_{\beta=1}^{N_B}\sum_{\ell=1}^{M_\beta^B}(p_{\beta,\ell}^B)^2}
=\sqrt{\mathcal{C}^A\mathcal{C}^B}
\leq\sqrt{\widetilde{\mathcal{C}}^A\widetilde{\mathcal{C}}^B}.
\end{equation}
Finally, from the convexity of the trace norm, this condition extends from product states to all separable states.
\end{proof}

The above results include as special cases the separability criteria for mutually unbiased bases \cite{Spengler,Liu} and measurements \cite{ChenMa,Liu2}, general SIC POVMs \cite{ChenLi}, $(N,M)$-POVMs \cite{SIC-MUB}, generalized symmetric measurements \cite{SIC-MUB_general}, and measurements based on equiangular tight frames \cite{W_ETS}.

An improved but more computationally demanding criterion is based on the matrix
\begin{equation}
\mathcal{R}_{\alpha,k;\beta,\ell}=\Tr\Big[(\rho-\rho_A\otimes\rho_B)
(P_{\alpha,k}^{A}\otimes P_{\beta,\ell}^{B})\Big],
\end{equation}
where $\rho^A=\Tr_B(\rho)$ and $\rho^B=\Tr_A(\rho)$ are the marginal states of $\rho$, obtained after calculating the partial trace over one subsystem.

\begin{Proposition}\label{esic2}
If a bipartite state $\rho$ is separable, then
\begin{equation}\label{enh}
\|\mathcal{R}\|_{{\rm{tr}}}\leq\sqrt{\widetilde{\mathcal{C}}^A-\mathcal{C}^A(\rho^A)}
\sqrt{\widetilde{\mathcal{C}}^B-\mathcal{C}^B(\rho^B)},
\end{equation}
where $\mathcal{C}^A(\rho^A)$ and $\mathcal{C}^B(\rho^B)$ are indices of coincidence for the marginal states of $\rho$.
\end{Proposition}

\begin{proof}
Any separable state $\rho$ can be written as
\begin{equation}
\rho=\sum_{j=1}^nq_j\rho^A_j\otimes\rho^B_j,
\end{equation}
where $q_j$ is a probability distribution. The corresponding marginal states read
\begin{equation}
\rho^A=\sum_{j=1}^nq_j\rho^A_j,\qquad \rho^B=\sum_{j=1}^nq_j\rho^B_j.
\end{equation}
Therefore, one finds
\begin{equation}\label{diff}
(\rho-\rho^A\otimes\rho^B)(P_{\alpha,k}^{A}\otimes P_{\beta,\ell}^{B})=\sum_{j=1}^nq_j\left(\rho^A_j-\rho^A\right)P_{\alpha,k}^{A}
\otimes\left(\rho^B_j-\rho^B\right)P_{\beta,\ell}^{B}.
\end{equation}
Now, introduce two vectors, $v_j^A$ and $v_j^B$, with the following elements,
\begin{equation}
\left(v^A_j\right)_{\alpha,k}=\Tr\left[\left(\rho^A_j-\rho^A\right)P_{\alpha,k}^{A}\right],\qquad 
\left(v^B_j\right)_{\beta,\ell}=\Tr\left[\left(\rho^B_j-\rho^B\right)P_{\beta,\ell}^{B}\right].
\end{equation}
This allows us to simplify the notation after taking the trace of eq. (\ref{diff}),
\begin{equation}
\mathcal{R}_{\alpha,k;\beta,\ell}=\Tr\Big[(\rho-\rho_A\otimes\rho_B)
(P_{\alpha,k}^{A}\otimes P_{\beta,\ell}^{B})\Big]=\sum_{j=1}^nq_j
\left(v^A_j\right)_{\alpha,k}\left(v^B_j\right)_{\beta,\ell}.
\end{equation}
Next, from the triangle inequality of the trace norm, we calculate
\begin{equation}\label{RR}
\|\mathcal{R}_{\alpha,k;\beta,\ell}\|_{\rm{tr}}\leq
\sum_{j=1}^nq_j\|v^A_j\|_{\rm{tr}}\|v^B_j\|_{\rm{tr}}
=\sum_{j=1}^nq_j\|v^A_j\|\|v^B_j\|,
\end{equation}
where we used the fact that the trace norm of a vector is simply a vector norm. From the Cauchy-Schwarz inequality,
\begin{equation}\label{CS}
\sum_{j=1}^nq_j\|v^A_j\|\|v^B_j\|\leq\sqrt{\sum_{j=1}^nq_j\|v^A_j\|^2}\sqrt{\sum_{m=1}^nq_m\|v^B_m\|^2}.
\end{equation}
In analogy to eq. (\ref{pak}), let us denote the trace expressions by
\begin{equation}
p_{\alpha,k}^A(\rho^A)=\Tr(P_{\alpha,k}^A\rho^A),\qquad 
p_{\beta,\ell}^B(\rho^B)=\Tr(P_{\beta,\ell}^B\rho^B),
\end{equation}
where the corresponding indices of coincidence are given by
\begin{equation}
\mathcal{C}^A(\rho^A)=\sum_{\alpha=1}^{N_A}\sum_{k=1}^{M_\alpha^A}
\left[p_{\alpha,k}^A(\rho^A)\right]^2,\qquad
\mathcal{C}^B(\rho^B)=\sum_{\beta=1}^{N_B}\sum_{\ell=1}^{M_\beta^B}
\left[p_{\beta,\ell}^B(\rho^B)\right]^2.
\end{equation}
Then, the term under the square root reduces to
\begin{equation}
\sum_{j=1}^nq_j\|v^A_j\|^2=\sum_{j=1}^nq_j\sum_{\alpha=1}^{N_A}\sum_{k=1}^{M_\alpha^A}
\Big(\Tr\left[\left(\rho^A_j-\rho^A\right)P_{\alpha,k}^{A}\right]\Big)^2
=\sum_{j=1}^nq_j\mathcal{C}^A(\rho_j^A)-\mathcal{C}^A(\rho^A)
\leq\widetilde{\mathcal{C}}^A-\mathcal{C}^A(\rho^A),
\end{equation}
with $\widetilde{\mathcal{C}}^A$ being the upper bound of $\mathcal{C}^A$,
and analogically for the subsystem $B$. After substituting these formulas into eq. (\ref{CS}) and then eq. (\ref{RR}), we reproduce eq. (\ref{enh}).
\end{proof}

This generalizes the separability criteria for popular measurements: mutually unbiased bases, mutually unbiased measurements, general SIC POVMs \cite{Liu2,Shen,Liu3}, and also for equiangular tight frames \cite{TangWu,QiZhang}. Moreover, our results allow to simplify these criteria significantly by proving that they depend only on the index of coincidence for the marginal state and its upper bound.

\begin{Example}
Consider a bipartite qutrit state $\rho=p\rho_x+(1-p)P_+$ that is a mixture of Horodecki's bound entangled state \cite{HORODECKI1997333}
\begin{equation}
\rho_x=\frac{1}{1+8x}
\left[\begin{array}{c c c|c c c|c c c}
x & \cdot & \cdot & \cdot & x & \cdot & \cdot & \cdot & x \\
\cdot & x & \cdot & \cdot & \cdot & \cdot & \cdot & \cdot & \cdot \\
\cdot & \cdot & x & \cdot & \cdot & \cdot & \cdot & \cdot & \cdot \\
\hline
\cdot & \cdot & \cdot & x & \cdot & \cdot & \cdot & \cdot & \cdot \\
x & \cdot & \cdot & \cdot & x & \cdot & \cdot & \cdot & x \\
\cdot & \cdot & \cdot & \cdot & \cdot & x & \cdot & \cdot & \cdot \\
\hline
\cdot & \cdot & \cdot & \cdot & \cdot & \cdot & \frac{1+x}{2} & \cdot & \frac{\sqrt{1-x^2}}{2} \\
\cdot & \cdot & \cdot & \cdot & \cdot & \cdot & \cdot & x & \cdot \\
x & \cdot & \cdot & \cdot & x & \cdot & \frac{\sqrt{1-x^2}}{2} & \cdot & \frac{1+x}{2}
\end{array}\right],
\end{equation}
$0<x<1$, and the maximally entangled state $P_+=(1/3)\sum_{m,n=1}^3|m\>\<n|\otimes|m\>\<n|$. Then, for this quantum state, construct the correlation matrix $\mathcal{P}$ with $P_{\alpha,k}^A=P_{\alpha,k}$ and $P_{\alpha,k}^B=\overline{P}_{\alpha,k}$, using the GEAM $\{P_{\alpha,k}\}$ given by eqs. (\ref{E12}--\ref{gamma1}). Numerical calculations show the range of parameter $p$ for which entanglement of $\rho$ is detected for a fixed $x$ -- that is, when the sufficient separability condition is broken. The results are presented in Table 1. Note that $\widetilde{C}$ is the upper bound for the index of coincidence corresponding to $P_{\alpha,k}$, and
\begin{equation}
\widetilde{C}_{AB}=\sqrt{\widetilde{\mathcal{C}}-\mathcal{C}(\rho^A)}
\sqrt{\widetilde{\mathcal{C}}-\mathcal{C}(\rho^B)}.
\end{equation}
Our results indicate that Proposition 7 detects the most and Proposition 5 the least entanglement. The differences in the range of $p$ decrease with the increase of the parameter $x$.
Comparing to analogical calculations in ref. \cite{Shen}, it is straightforward to see that our choice of a GEAM is able to detect more entanglement than MUMs and general SIC POVMs.
\end{Example}

\begin{table}[ht]
\centering
\begin{tabular}{|c|c|c|}
\hline
Value of $x$ & Entanglement criterion & Range of $p$ \\
\hline
\multirow{3}{*}{$x=0.1250$} & $\Tr\mathcal{P}>\widetilde{\mathcal{C}}$ & $0\leq p\leq 0.9014$ \\
& $\|\mathcal{P}\|_{\rm tr}>\widetilde{\mathcal{C}}$ & $0\leq p\leq 0.9310$ \\
& $\|\mathcal{R}\|_{\rm tr}>\widetilde{\mathcal{C}}_{AB}$ & $0\leq p\leq 0.9605$ \\
\hline
\multirow{3}{*}{$x=0.3750$} & $\Tr\mathcal{P}>\widetilde{\mathcal{C}}$ & $0\leq p\leq 0.9624$ \\
& $\|\mathcal{P}\|_{\rm tr}>\widetilde{\mathcal{C}}$ & $0\leq p\leq 0.9722$ \\
& $\|\mathcal{R}\|_{\rm tr}>\widetilde{\mathcal{C}}_{AB}$ & $0\leq p\leq 0.9791$ \\
\hline
\multirow{3}{*}{$x=0.6250$} & $\Tr\mathcal{P}>\widetilde{\mathcal{C}}$ & $0\leq p\leq 0.9846$ \\
& $\|\mathcal{P}\|_{\rm tr}>\widetilde{\mathcal{C}}$ & $0\leq p\leq 0.9884$ \\
& $\|\mathcal{R}\|_{\rm tr}>\widetilde{\mathcal{C}}_{AB}$ & $0\leq p\leq 0.9906$ \\
\hline
\multirow{3}{*}{$x=0.8750$} & $\Tr\mathcal{P}>\widetilde{\mathcal{C}}$ & $0\leq p\leq 0.9961$ \\
& $\|\mathcal{P}\|_{\rm tr}>\widetilde{\mathcal{C}}$ & $0\leq p\leq 0.9970$ \\
& $\|\mathcal{R}\|_{\rm tr}>\widetilde{\mathcal{C}}_{AB}$ & $0\leq p\leq 0.9975$ \\
\hline
\end{tabular}
\caption{Range of the parameter $p$ for which entanglement of $\rho$ is detected by the criteria from Propositions 5--7.}
\label{table1}
\end{table}

\section{Conclusions}

In this paper, we applied the recently introduced generalized equiangular measurements to quantum entanglement detection. It was shown that, as long as GEAMs are conical 2-designs, the purity of the state is directly related to the partial index of coincidence, where the summation of squared probabilities is performed over few generalized equiangular tight frames. Such measurements were then used to define a class of positive but not trace-preserving maps and the corresponding entanglement witnesses. Interestingly, these witnesses are parametrized solely by two partial indices of coincidence. Examples for bound entanglement detection had been provided. Next, we used the generalized equiangular measurements to construct separability criteria for bipartite states with subsystems of arbitrary dimensions.

In further research, it would be interesting to examine the properties of entanglement witnesses, such as extremality or optimality. Construction of witnesses and separability criteria from measurements that are not conical 2-designs remains as an open question. Another possible direction is to generalize the presented formalism to multipartite systems and subsystems of unequal dimensionality.

\section{Acknowledgements}

This research was funded in whole or in part by the National Science Centre, Poland, Grant number 2021/43/D/ST2/00102. For the purpose of Open Access, the author has applied a CC-BY public copyright licence to any Author Accepted Manuscript (AAM) version arising from this submission.

\section{Author Contribution}

K.S. is solely responsible for the entire research process.

\bibliography{bibliography}
\bibliographystyle{unsrt}

%

\end{document}